\documentclass[12pt]{amsart}

\usepackage{amssymb,amsthm,amsmath}
\usepackage[numbers,sort&compress]{natbib}
\usepackage{color}
\usepackage{graphicx}
\usepackage{color}
\usepackage{type1cm}  
\title[resonant embedded eigenvalues]{One dimensional discrete Schr\"odinger operators with resonant embedded   eigenvalues }

\author{Wencai Liu}
\address{ Department of Mathematics, Texas A\&M University, College Station, TX 77843-3368, USA}
\email{liuwencai1226@gmail.com; wencail@tamu.edu}
\author{Kang Lyu}
\address{ Department of Applied Mathematics, Nanjing University of Science and Technology, Nanjing, 210094, Jiangsu, People’s Republic of China}
\email{lvkang201905@outlook.com}

\newcommand{\abs}[1]{\left\lvert #1 \right\rvert}

\theoremstyle{plain}
\newtheorem{theorem}{Theorem}[section]

\newtheorem{lemma}[theorem]{Lemma}
\newtheorem{proposition}[theorem]{Proposition}

\newcommand{\R}{\mathbb{R}}
\newcommand{\Z}{\mathbb{Z}}
\newcommand{\N}{\mathbb{N}}
\theoremstyle{definition}

\begin{document}

	%%% ----------------------------------------------------------------------

	%%% ----------------------------------------------------------------------
	\maketitle

	%%% ----------------------------------------------------------------------
	%%%%%%%%%%%%%%%%%%%%%%%%%%%%%%%%%%%%%%%%%%%%%%%%%%%%%%%%%%%%%%%%%%%%%%%%%%
	%INTRODUCTION%%%%%%%%%%%%%%%%%%%%%%%%%%%%%%%%%%%%%%%%%%%%%%%%%%%%%%%%%%%%%
	%%%%%%%%%%%%%%%%%%%%%%%%%%%%%%%%%%%%%%%%%%%%%%%%%%%%%%%%%%%%%%%%%%%%%%%%%%
\ \ \ \ {\itshape 	Dedicated to the memory of Sergey Nikolaevich Naboko (1950--2020)}
	\begin{abstract}
In this paper, we introduce a new family of functions to construct Schr\"odinger operators with embedded eigenvalues.
This particularly  allows   us to construct discrete Schr\"odinger operators with arbitrary prescribed sets of eigenvalues. 

\end{abstract}
\section{Introduction}
In this paper, we study  the discrete Schr\"odinger operator,
\begin{align}
Hu(n)=(\Delta+V)u(n)=u(n+1)+u(n-1)+V(n)u(n),\nonumber
\end{align}
and the eigen-equation $Hu=Eu$, where    $\Delta$   is the free discrete Schr\"odinger operator, and $V$ is the potential.

For simplicity, we only consider the eigen-equation on  the half line,  
\begin{align}\label{mainequation}
Hu(n)=u(n+1)+u(n-1)+V(n)u(n)=Eu(n),\ \ n\geq1,
\end{align}
with the boundary condition
\begin{align}\label{boundarycondition}
\frac{u(1)}{u(0)}=\tan\theta,\theta\in [0,\pi).
\end{align}

By  Weyl's criterion for the essential
	spectrum, when $V(n)=o(1)$ as $n\to \infty$,  the essential spectrum of $H$ $$\sigma_{ess}(H)=\sigma_{ess}(\Delta)=[-2,2].$$
We are interested in constructions of potentials $V$ such that $H=\Delta+V$ have 
eigenvalues embedded into the essential spectrum  $[-2,2]$.

The study of problems of embedded eigenvalues has a long history. The classical Wigner-von Neumann functions
$
V_k(x)=\frac{K}{1+x}\sin(kx+\phi)
$ could provide one embedded eigenvalue. Naboko~\cite{Na86} and Simon~\cite{simondense} constructed continuous Schr\"odinger operators with countably many (dense) eigenvalues. Recently there are remarkable developments in the area of embedded eigenvalues for both continuous and discrete Schr\"odinger operators  ~\cite{sim07,liupafa,kru,liu21,lotoreichik2014spectral,lukwn,luk14,lukd1,KRS,sim12,remling1998absolutely,jnw1,jnw2,jnw3}.

The previous results about perturbed discrete Schr\"odinger operators can be 
summarized as follows.
For any $E\in(-2,2)$, let $E=2\cos \pi k(E)$ with $k(E)\in(0,1)$ ($k(E)$ is referred to as quasimomentum). 
Naboko and Yakovlev ~\cite{naboko1992point} constructed potentials $V$ such that $\Delta+V$ have any given eigenvalues that their quasimomenta are rationally independent (see ~\cite{Na86,KRS,simondense} for continuous cases and ~\cite{napu,L19stark} for Stark type potentials). In ~\cite{JL19}, Jitomirskaya and Liu introduced the piecewise constructions and gluing techniques to construct embedded eigenvalues of Laplacians on non-compact manifolds. Those techniques have been further developed by Liu and Ong in ~\cite{LO17}. In particular, they constructed potentials of any prescribed embedded eigenvalues when their quasimomenta are non-resonant, namely $k(E)+k(\tilde{E})\neq 1$ (in other words, $E\neq -\tilde{E}$). See ~\cite{l1} for an alternative construction based on explicit Wigner-von Neumann type functions.
% The first author has fully expanded the approach by Jitomirskaya and Liu to study various models. 

It is natural to ask if we could relax the non-resonance assumption in ~\cite{l1,LO17}, namely construct discrete Schr\"odinger operators with arbitrary prescribed sets of eigenvalues. 

When $k(E)=1-\tilde{k}(E)$, their corresponding Wigner-von Neumann functions $V_k(x)=\frac{K}{1+x}\sin(kx+\phi) $ and $V_{\tilde k}(x)=\frac{K}{1+x}\sin(\tilde k x+\phi) $ are the same up to a phase $\phi$. This is a 
big obstacle to construct potentials such that both $E$ and $\tilde{E}$ are eigenvalues. We also want to mention that non-resonant conditions appear naturally in the study of Schr\"odinger operators with Wigner-von Neumann type or bounded variation potentials ~\cite{lukd1,JanasSimonov}.

The main goal of this paper is to construct arbitrary embedded eigenvalues (without non-resonance restriction).
The proof (contains two steps) follows the general scheme of Jitomirskaya and Liu ~\cite{JL19}. The first step is to construct a ``nice" function (we call it a generating function) attached to one chosen eigenvalue $E$. The second step is to glue all pieces together by adding eigenvalues slowly (the idea is inspired by Naboko ~\cite{Na86} and Simon ~\cite{simondense}). 
With the recent development ~\cite{LO17,L19stark,l1}, the second step is rather standard.
So following the plan of Jitomirskaya and Liu ~\cite{JL19}, the difficulty lies in finding the generating function. 
The Wigner-von Neumann function $V_k(x)=\frac{K}{1+x}\sin(kx+\phi)
$ ($k$ is the quasimomentum of $E$) is a natural candidate, which has been used to do the constructions in ~\cite{JL19} and ~\cite{l1}. 
In order to study perturbed periodic operators, inspired by ~\cite{KRS}, Liu and Ong ~\cite{LO17} chose the generating function by solving an equation involving  Pr\"ufer angles and potentials. Both generating functions fail in the case of resonant eigenvalues, since 
the generating function attached to $E$ completely breaks  the oscillatory estimate to the eigen-equation $Hu=\tilde{E} u$ (for example, see Lemma \ref{keyle}).

In the present work, we propose a new family of functions as our generating functions.
There are two novelties in our constructions of generating functions. We first obtain a function $\tilde V_k+\tilde V_{\tilde{k}}$ attached to a pair of resonant eigenvalues $E$ and $\tilde{E}$ by solving a system of two equations. Both functions $\tilde V_k$ and $\tilde V_{\tilde{k}}$ have similar properties of (but not) Wigner-von Neumann functions. However, $\tilde V_k+\tilde V_{\tilde{k}}$ can not directly serve as a generating function since the Pr\"ufer angle $\theta_{\tilde{k}}(n)+{\theta}_k(n)$ may stay fixed up to an integer by the fact that $k+\tilde{k}=1$.
In order to overcome the difficulty, 
we modify the function $\tilde V_k+\tilde V_{\tilde{k}}$ by  adding an extra function $\frac{1}{1+n}$. In other words, our generating functions have the format $\tilde V_k+\tilde V_{\tilde{k}}+\frac{1}{1+n}$. The additional function $\frac{1}{1+n}$ will play a significant role in our constructions. We believe the novel family of functions $\tilde V_k+\tilde V_{\tilde{k}}+\frac{1}{1+n}$ introduced in this paper will have wider applications, in particular problems of perturbed periodic operators with resonant embedded eigenvalues. 
	
	\begin{theorem}\label{theorem1}
		
		 Suppose   $\{E_j\}_{j=1}^m $  are    a finite set of distinct points in  $(-2,2)$.
Then for any given  $\{\theta_j\}_{j=1}^m\subset[0,\pi)$, there exist potentials $V(n)=\frac{O(1)}{1+n}$ such that for any $j=1,2,\cdots,m$, the eigen-equation $Hu=(\Delta+V)u=E_ju$ has a solution $u\in l^2(\mathbb{N})$ satisfying the boundary condition
		\begin{align}
			\frac{u(1)}{u(0)}=\tan\theta_j.\nonumber
		\end{align}
	\end{theorem}
	
	\begin{theorem}\label{theorem2}
		Suppose $ \{E_j\}_{j=1}^{\infty}$  are   a countable set of distinct points in  $(-2,2)$.
Then for any given   $\{\theta_j\}_{j=1}^{\infty}\subset[0,\pi)$ and any non-decreasing positive function $h(n)$ with $\lim_{n\to\infty}h(n)=\infty$,
		there exist potentials $V$ satisfying
		$$\abs{V(n)}\leq\frac{h(n)}{1+n},$$
		such that for any $j=1,2,\cdots$, the eigen-equation $Hu=(\Delta+V)u=E_ju$ has a solution $u\in l^2(\mathbb{N})$ satisfying the boundary condition
		\begin{align}
			\frac{u(1)}{u(0)}=\tan\theta_j.\nonumber
		\end{align}
	\end{theorem}

	\section{some basic lemmas}
	The basic tool we use in this paper is the  modified Pr\"ufer transformation. We refer readers to  ~\cite{Simon1998EFGP} and ~\cite{remling1998absolutely} for more details. 
	
	For any $E\in (-2,2)$, denote by $E=2\cos \pi k$ with $k=k(E)\in (0,1)$. Suppose $u(n)=u(n,E)$ is a solution of \eqref{mainequation}. Let
	\begin{align}
		Y(n)=\frac{1}{\sin\pi k}\begin{pmatrix}
			\sin \pi k &  0\\
			-\cos \pi k & 1
		\end{pmatrix}
		\begin{pmatrix}
			u(n-1)\\
			u(n)
		\end{pmatrix}.\nonumber
	\end{align}
	
	Define the Pr\"ufer variables $R(n)$ and $\theta(n)$ as 
	\begin{align}
		Y(n)=R(n)\begin{pmatrix}
			\sin (\pi\theta(n)-\pi k)\\
			\cos (\pi\theta(n)-\pi k)
		\end{pmatrix}.\nonumber
	\end{align}
	Then $R(n)$ and $\theta(n)$ obey
	\begin{eqnarray}\label{lr1}
		\frac{R(n+1)^2}{ R(n)^2}=1-\frac{V(n)}{\sin\pi k}\sin2\pi\theta(n)+\frac{V(n)^2}{\sin^2\pi k}\sin^2\pi\theta(n)
	\end{eqnarray}
	and
	\begin{eqnarray}\label{ctt}
		\cot (\pi \theta(n+1)-\pi k)=\cot \pi \theta(n)-\frac{V(n)}{\sin\pi k}.
	\end{eqnarray}
	
	When $|\frac{V(n)}{\sin\pi k}|\leq \frac{1}{10}$, by \eqref{lr1}, one has that 
		\begin{eqnarray}\label{lr}
\ln  R(n+1)^2-\ln R(n)^2 =-\frac{V(n)}{\sin\pi k}\sin2\pi\theta(n)+O\left(\frac{V(n)^2}{\sin^2\pi k} \right)
	\end{eqnarray}
	
	Now let us introduce some useful lemmas.
	\begin{lemma}~\cite[Proposition 2.4]{Simon1998EFGP} Suppose $\abs{\frac{V(n)}{\sin\pi k}}<\frac{1}{2}$, then the $\theta(n)$ defined by \eqref{ctt} satisfies
		\begin{align}\label{bhgx1}
			\abs{\theta(n+1)-\theta(n)-k}\leq\abs{\frac{V(n)}{\sin\pi k}}.
		\end{align}
	\end{lemma}

\begin{lemma}
	Suppose $\abs{\frac{V(n)}{\sin\pi k}}< \frac{1}{10}$, then the $\theta(n)$ defined by \eqref{ctt} satisfies 
	\begin{align}\label{bhgx2}
	\theta(n+1)=\theta(n)+k+\sin^2\pi\theta(n)\frac{V(n)}{\pi\sin\pi
		k}+O\left(\frac{V(n)^2}{\sin^2\pi k}\right).
	\end{align}
\end{lemma}
\begin{proof}
	Denote by $\theta_1=\pi\theta(n+1)-\pi k,\ \theta_0=\pi\theta(n)$.
	
	Since
	\begin{align}
	e^{2i\theta}=1-\frac{2}{1+i\cot\theta},\nonumber
	\end{align}
	by \eqref{ctt} we have
	\begin{align}
	e^{2i\theta_1}=&e^{2i\theta_0}-\frac{\frac{iV(n)}{2\sin\pi k}(e^{2i\theta_0}-1)^2}{1-\frac{iV(n)}{2\sin\pi k}(1-e^{2i\theta_0})}\nonumber\\
	=&e^{2i\theta_0}-{\frac{iV(n)}{2\sin\pi k}(e^{2i\theta_0}-1)^2}\left(1+O\left(\frac{V(n)}{\sin\pi k}\right)\right)\nonumber\\
	=&e^{2i\theta_0}\left(1-\frac{iV(n)}{2\sin\pi k}e^{2i\theta_0}-\frac{iV(n)}{2\sin\pi k}e^{-2i\theta_0}+\frac{iV(n)}{\sin\pi k}\right)+O\left(\frac{V(n)^2}{\sin^2\pi k}\right)\nonumber\\
	=&e^{2i\theta_0}\left(1-\frac{iV(n)}{\sin\pi k}\cos2\theta_0+\frac{iV(n)}{\sin\pi k}\right)+O\left(\frac{V(n)^2}{\sin^2\pi k}\right)\nonumber\\
	=&e^{2i\theta_0}\left(1+\frac{2iV(n)}{\sin\pi k}\sin^2\theta_0\right)+O\left(\frac{V(n)^2}{\sin^2\pi k}\right)\nonumber\\
	\label{ee}=&e^{2i\theta_0}e^{2i\frac{V(n)}{\sin\pi k}\sin^2\theta_0}+O\left(\frac{V(n)^2}{\sin^2\pi k}\right).
	\end{align}
	Since $\abs{e^{i\theta}-1}\geq\frac{2}{\pi}\abs{\theta}$, by \eqref{ee} one obtains
	\begin{align}
	\frac{4}{\pi}\abs{\theta_1-\theta_0-\frac{V(n)}{\sin\pi k}\sin^2\theta_0}\leq\abs{e^{2i\theta_1}-e^{2i\left(\theta_0+\frac{V(n)}{\sin\pi k}\sin^2\theta_0\right)}}=O\left(\frac{V(n)^2}{\sin^2\pi k}\right).\nonumber
	\end{align}
This leads to \eqref{bhgx2}.
	
\end{proof}
	\begin{lemma}\label{ki}
		Let $ k\in (0,1)$ with $k\neq \frac{1}{2}$.  Then for any $\varepsilon>0$,  there exists   $N\in \mathbb{N}$ (depending on $k$ and $\varepsilon>0$)  such  that for any $\theta \in  \R$ and   $\nu\in\{2,4\}$, we have
		\begin{align}\label{cos1}
			\abs{\frac{1}{N}\sum_{l=0}^{N-1}\cos(\theta\pm\nu\pi kl)}\leq \varepsilon,
		\end{align}
		\begin{align}\label{sin1}
			\abs{\frac{1}{N}\sum_{l=0}^{N-1}\sin(\theta\pm\nu\pi k l)}\leq \varepsilon.
		\end{align}
	\end{lemma}
	\begin{proof}
		Clearly,  \eqref{sin1} follows from \eqref{cos1}.  In order to avoid  repetition, we only prove \eqref{cos1} for $\nu=2$.
%		We only show that for any large enough $q$, we have
%		\begin{align}
%			\abs{\frac{1}{q}\sum_{j=0}^{q-1}\cos(\theta+2\pi k j)}\leq \frac{\sin\pi k}{10^{10}},\nonumber
%		\end{align}
%		the proofs of other cases are similar. 
		
		{\bf Case 1}: $k$ is rational
		
		In this case, let $k=\frac{p}{q}$,  and $p$ and $q$ are coprime.
		Direct computations imply that for $x\notin 2\pi \Z$, 
		\begin{align}
			\cos \theta+\cos(\theta+x)+\cdots+\cos(\theta+(N-1)x)=\frac{\sin\left(\theta-\frac{x}{2}+Nx\right)-\sin\left(\theta-\frac{x}{2}\right)}{2\sin\frac{x}{2}}.\nonumber
		\end{align}
	Therefore, one has that for $N=q$, 
		\begin{align}
			\abs{\frac{1}{N}\sum_{l=0}^{N-1}\cos(\theta+2\pi k l)}=0.\nonumber
		\end{align}
		This obviously implies \eqref{cos1}.
		
			{\bf Case 2}: $k$ is irrational
			
			We note that the map $T:x\to x+2k\pi \mod 2\pi$ is ergodic.
			By the ergodicity, one has that for large $N$, 
		\begin{align}
			\abs{\frac{1}{N}\sum_{l=0}^{N-1}\cos(\theta+2\pi k  l)}=	\abs{\frac{1}{N}\sum_{l=0}^{N-1}\cos(\theta+2\pi k  l)-\int_{0}^{2\pi} \cos tdt }\leq\varepsilon.\nonumber
		\end{align}
		
	We finish the proof.
\end{proof}
	\section{Technical preparations}
	Let $N\in \N$, which will be determined later.
\begin{lemma}\label{zhongyaoyinli}
	Let  $ k\in\left(0,1\right)$,    $L\in\N$, $n_0\in\N$ and  $b \in\mathbb{N}$.  Assume that  $F=F(n)$ and $\varphi=\varphi(n)$ satisfy for any $m\in\mathbb{N}$,
	\begin{align}\label{llrnmx}
	\ln F(n_0+(m+1)&N)^2=\ln F(n_0+mN)^2\\
	-&\frac{K}{2(n_0+mN-b)\sin\pi k}(1-\cos2\pi\varphi(n_0+mN)+\delta(m))\nonumber
	\end{align}
	and 
	\begin{align}\label{ttnmx}
	\varphi(n_0+(m+1)N)=\varphi(n_0+mN)+L+\frac{K}{(n_0+mN-b)\pi\sin\pi k}\left(100+\varepsilon(m)\right) ,
	\end{align}
	where $K>10^8\pi N$, $\abs{\delta(m)}\leq \frac{\sin\pi k}{10^4}$ and $\abs{\varepsilon(m)}<1$. Then we have that for large enough $n_0-b>0$,
	\begin{align}\label{rn0jian}
	F(n_0+mN)\leq 1.5F(n_0)\left(\frac{n_0-b+mN}{n_0-b}\right)^{-100}.
	\end{align}
\end{lemma}

\begin{proof}[\textbf{Proof.}]
	Denote by 
 	$\tilde\varphi(n_0+mN)=\varphi(n_0+mN)-mL.$
Then by \eqref{llrnmx} and \eqref{ttnmx} one has that
\begin{align}\label{llrnm}
	\ln F(n_0+(m+1)&N)^2=\ln F(n_0+mN)^2\\
	-&\frac{K}{2(n_0+mN-b)\sin\pi k}(1-\cos2\pi\tilde\varphi(n_0+mN)+\delta(m))\nonumber
\end{align}
and
\begin{align}\label{ttnm}
\tilde\varphi(n_0+(m+1)N)=\tilde\varphi(n_0+mN)+\frac{K}{(n_0+mN-b)\pi\sin\pi k}\left(100+\varepsilon(m)\right).
\end{align}

	Inequality 
$1-\cos2\pi\tilde\varphi(n_0+mN)+\delta(m)\geq -\frac{\sin\pi k}{10^4}$
 and \eqref{llrnm} imply that 
	\begin{align}\label{zenda}
	\ln F(n_0+(m+1)N)^2\leq \ln F(n_0+mN)^2+\frac{K}{2\times10^4}\frac{1}{n_0+mN-b}.
	\end{align}
	
In the case that $\tilde\varphi(n_0+mN)\mod\Z\in\left[\frac{1}{3},\frac{2}{3}\right]$,   one has that 
	\begin{align}
	1-\cos2\pi\tilde\varphi(n_0+mN)+\delta(m)\geq 1,\nonumber
	\end{align}
	and hence (by \eqref{llrnm}),
	\begin{align}\label{jianxiao}
	\ln F(n_0+(m+1)N)^2\leq \ln F(n_0+mN)^2-\frac{K}{2\sin\pi k}\frac{1}{n_0+mN-b}.
	\end{align}

Without loss of generality, we assume that $\tilde\varphi(n_0)\in \left[-\frac{4}{3},-\frac{1}{3}\right]$.  By \eqref{ttnm} we  have that there exists  an increasing sequence $\{m_l\}_{l=1}^{\infty}\subset \mathbb{N}$ such that for any $l\in\N$,
		\begin{align}\label{gj261}
	\tilde\varphi(n_0+m_{2l}N)\leq\frac{3l-1}{3},\ \ \tilde\varphi(n_0+(m_{2l}+1)N)>\frac{3l-1}{3}, 
	\end{align}
	and 
	\begin{align}\label{gj262}
	\tilde\varphi(n_0+m_{2l+1}N)\leq\frac{3l+1}{3},\ \  \tilde\varphi(n_0+(m_{2l+1}+1)N)>\frac{3l+1}{3}. 
	\end{align}

By \eqref{ttnm}, one has that 
			 \begin{align}
				&\tilde\varphi(n_0+(m_{2l+1}+1)N)\nonumber\\
				=&\tilde\varphi(n_0+m_{2l}N)+\sum_{m=m_{2l}}^{m_{2l+1}}\frac{K}{(n_0+mN-b)\pi\sin\pi k}(100+\varepsilon(m)). \label{gj264}
	\end{align}
Inequalities  \eqref{gj261} and \eqref{gj262} imply that 
\begin{equation}\label{gj263}
\tilde\varphi(n_0+(m_{2l+1}+1)N)-\tilde\varphi(n_0+m_{2l}N)\geq  	\frac{2}{3}.
\end{equation}
By \eqref{gj264} and \eqref{gj263}, we have that 
\begin{align}
	\frac{2}{3}\leq&\sum_{m=m_{2l}}^{m_{2l+1}}\frac{K}{(n_0+mN-b)\pi\sin\pi k}(100+\varepsilon(m))\nonumber\\
	\leq &101\frac{K}{N\pi\sin\pi k}\sum_{m=m_{2l}}^{m_{2l+1}}\frac{1}{\frac{n_0-b}{N}+m}\nonumber\\
	\leq &102C_k\ln\frac{m_{2l+1}N+n_0-b}{m_{2l}N+n_0-b},\label{gj266}
\end{align}
where $C_k=\frac{K}{N\pi\sin\pi k}\geq 10^8$.

%Replace $m_{2l+1}+1$ and $m_{2l}$  with  $m_{2l+1}$ and $m_{2l}+1$ respectively. 
Replacing \eqref{gj263}  with  the  inequality 
\begin{align}
	\tilde\varphi(n_0+m_{2l+1}N)-\tilde\varphi(n_0+(m_{2l}+1))N)\leq \frac{2}{3}\nonumber
\end{align}
and following  the proof  of \eqref{gj266}, 
we have
\begin{align}
	98C_k\ln \frac{m_{2l+1}N+n_0-b}{m_{2l}N+n_0-b}\leq \frac{2}{3}.\nonumber
\end{align}

Therefore, one  concludes  that for any $l\in\N$, 
	\begin{align}\label{gnew11}
	98C_k\ln \frac{m_{2l+1}N+n_0-b}{m_{2l}N+n_0-b}\leq \frac{2}{3}\leq 102C_k\ln\frac{m_{2l+1}N+n_0-b}{m_{2l}N+n_0-b}. 
	\end{align}
 Similar  to the proof of  \eqref{gnew11}, we have that  for any $l\in\N$,
	\begin{align}\label{gnew21}
	98C_k\ln \frac{m_{2l+2}N+n_0-b}{m_{2l+1}N+n_0-b}\leq \frac{1}{3}\leq 102C_k\ln\frac{m_{2l+2}N+n_0-b}{m_{2l+1}N+n_0-b}.
	\end{align}

	Rewrite \eqref{gnew11} and \eqref{gnew21} as
		\begin{align}\label{gnew1l}
 \frac{1}{153 C_k}\leq  \ln \frac{m_{2l+1}N+n_0-b}{m_{2l}N+n_0-b}\leq \frac{1}{147 C_k},
	\end{align}
	and 
	\begin{align}\label{gnew2l}
 \frac{1}{306 C_k}\leq \ln \frac{m_{2l+2}N+n_0-b}{m_{2l+1}N+n_0-b}\leq  \frac{1}{294 C_k}.
	\end{align}

	By \eqref{zenda} and \eqref{jianxiao}, we obtain for any    $l\in\N$,
	\begin{align} \label{gnew3}
	\ln F(n_0+m_{2l+1}N)^2\leq& \ln F(n_0+m_{2l}N)^2+\frac{K}{10^4N}\ln\frac{m_{2l+1}N+n_0-b}{m_{2l}N+n_0-b}\nonumber\\
	&+O\left(\frac{1}{m_{2l}N+n_0-b}\right),
	\end{align}
	and
	\begin{align}\label{gnew4}
	\ln F(n_0+m_{2l+2}N)^2\leq &\ln F(n_0+m_{2l+1}N)^2-\frac{K}{2N\sin\pi k}\ln \frac{m_{2l+2}N+n_0-b}{m_{2l+1}N+n_0-b} \nonumber\\&+O\left(\frac{1}{m_{2l+1}N+n_0-b}\right). 
	\end{align}
	By \eqref{gnew1l} and \eqref{gnew3}, one has  for large $n_0-b$, 
%By \eqref{n0n1n0}, we have
	\begin{align}\label{rn1rn0}
	F(n_0+m_{2l+1}N)^2\leq F(n_0+m_{2l}N)^2e^{\frac{K}{146NC_k\times10^4}}\leq F(n_0+m_{2l}N)^2e^{\frac{1}{10^5}}.
	\end{align}
By \eqref{gnew2l} and \eqref{gnew4}, one has  for large $n_0-b$, 
	\begin{align}\label{rn2rn1}
	F(n_0+m_{2l+2}N)^2\leq F(n_0+m_{2l+1}N)^2e^{-\frac{K}{613NC_k\sin\pi k}}\leq F(n_0+m_{2l}N)^2e^{-\frac{1}{10^3}}.
	\end{align}
	
	Moreover, it is not difficult to see that (similar to the proof of \eqref{rn1rn0})  for any      $l\in\N$ and $m\in [m_{2l},m_{2l+2}]$,
	\begin{align}\label{Rmax}
	F(n_0+mN)^2\leq F(n_0+m_{2l}N)^2e^{\frac{1}{10^5}}.
	\end{align}

	Iterating \eqref{rn2rn1} ($l$ times) and by \eqref{Rmax}, one has that for any  $l\in \N$ and $m\in[m_{2l},m_{2l+2}]$,
	
 	\begin{align} 
	\label{R2l1}	F(n_0+mN)^2\leq F(n_0+m_{0}N)^2e^{-\frac{l}{10^3}+\frac{1}{10^5}}.
	\end{align}
	
 Similar to the proof  of  \eqref{gnew11} or \eqref{gnew21}, one has that
 \begin{align}\label{gnew10}
  \ln \frac{m_0N+n_0-b}{n_0-b}\leq \frac{1}{98C_k}.
 \end{align}
	Similar to the proof  of \eqref{rn1rn0} (also see \eqref{Rmax}),
	one has that 
\begin{equation}\label{g18}
		F(n_0+m_0N)^2\leq e^{\frac{1}{10^5}}F(n_0)^2.
	\end{equation}
	By \eqref{R2l1} and \eqref{g18}, 
	one has that  for any    $l\in\N$ and $m\in[m_{2l},m_{2l+2}]$,
 	
	\begin{align} 
	\label{R2l11}	F(n_0+mN)^2\leq F(n_0)^2 e^{-\frac{l}{10^3}+\frac{2}{10^5}}.
	\end{align}
By \eqref{gnew1l} and \eqref{gnew2l}, one has 
	\begin{align}\label{gnew2x}
	\frac{l+1}{102C_k}\leq\ln \frac{m_{2l+2}N+n_0-b}{m_0N+n_0-b}\leq \frac{l+1}{98C_k}.
	\end{align}	 
By \eqref{gnew10} and \eqref{gnew2x}, one has that
	\begin{align}
		\ln\frac{m_{2l+2}N+n_0-b}{n_0-b}=\ln\frac{m_{2l+2}N+n_0-b}{m_0N+n_0-b}+\ln\frac{m_{0}N+n_0-b}{n_0-b}
		\leq {\frac{l+2}{98C_k}}.\nonumber
	\end{align}
This implies that
\begin{align}\label{gnew3x}
	l\geq 98C_k\ln\frac{m_{2l+2}N+n_0-b}{n_0-b}-2\geq98C_k\ln\frac{mN+n_0-b}{n_0-b}-2.
\end{align}
By \eqref{R2l11}, \eqref{gnew3x} and the fact that $C_k\geq 10^8$, one has 
\begin{align}
	F(n_0+mN)^2\leq& F(n_0)^2e^{\frac{2}{10^5}+\frac{2}{10^3}}e^{-\frac{98C_k}{10^3}\ln\frac{mN+n_0-b}{n_0-b}}\nonumber\\
	\leq& F(n_0)^2e^{\frac{3}{10^3}}\left(\frac{mN+n_0-b}{n_0-b}\right)^{-200}.\nonumber
\end{align}
This implies  \eqref{rn0jian}.

%	Then we can obtain \eqref{rn0jian} by \eqref{m0qn}.
\end{proof}

	\section{Constructions of  potentials and proof of Theorems \ref{theorem1} and \ref{theorem2} }
	
	In this section, we always assume that $E=2\cos\pi k, E_j=2\cos \pi k_j$, $\tilde{E}=-E=2\cos\pi \tilde{k},$ we also denote by $(R(n),\theta(n)), (R_j(n),\theta_j(n)), (\tilde{R}(n),\tilde{\theta}(n))$ the Pr\"ufer variables of $E,E_j,\tilde{E}$ (or say $k,k_j,\tilde{k}$). We remark that $k=1-\tilde{k}$ and 
	 $
		\sin\pi k=\sin\pi \tilde{k}.
	$

%For any $0\neq E\in (0,2)$ ($\frac{1}{2}\neq k\in (0,1)$), we let $q$ be the smallest positive integer in the Lemma \ref{ki} so that we have \eqref{cos1} and \eqref{sin1}.
	\begin{proposition}\label{prop}
		Let $A=\{E_j\}_{j=1}^m\subset(-2,2)$. 
		Let $E\in(-2,2)$  be such that  both $E$ and $\tilde{E}=-E$ are not in $A$. 	Suppose $\theta_0,\tilde{\theta}_0\in [0,\pi)$.   Then there exist constants $K_1(E,A),K_2(E,A)$  and  a function  $V(E,A,n_0,b,\theta_0,\tilde{\theta}_0)$ such that  the following holds for $n_0-b\geq K_2(E,A)$:
	    \begin{description}
	    	\item[\textbf{Perturbation}] ${\rm{supp}}(V)\subset[n_0,\infty)$, and for any $n\geq n_0$,
	    	\begin{align}\label{vxc}
	    		\abs{V(E,A,n_0,b,\theta_0,\tilde{\theta}_0)}\leq \frac{K_1(E,A)}{n-b};
	    	\end{align}
	    	\item[Solution for $E$] the solution of $(\Delta+V)u=Eu$ with the boundary condition $\theta(n_0)=\theta_0$ satisfies for any $n$ with $n>n_0$,
	    	\begin{align}\label{R-100}
	R(n)\leq 2\left(\frac{n-b}{n_0-b}\right)^{-100}R(n_0);
\end{align}

\item[Solution for $\tilde{E}$ (if $E\neq 0$)] the solution of $(\Delta+V)u=\tilde{E}u$ with the boundary condition $\tilde{\theta}(n_0)=\tilde{\theta}_0$ satisfies for any $n>n_0$,
\begin{align}\label{Rtilde-100}
	\tilde{R}(n)\leq 2\left(\frac{n-b}{n_0-b}\right)^{-100}\tilde{R}(n_0);
\end{align}

\item[Solution for $E_j$] any solution of $(\Delta+V)u=E_ju$ satisfies  for any $n>n_0$, 
\begin{align}\label{RjC}
	R_j(n)\leq 2R_j(n_0).
\end{align}
\end{description}
		\end{proposition}
	
The proof of the following lemma is similar to that of  ~\cite[Lemma 3.1]{liuabsence}. For readers' convenience, we include a proof.
\begin{lemma}\label{keyle}
	Let $E$ and $A=\{E_j\}_{j=1}^m$ satisfy the same assumptions as in Proposition \ref{prop}. Assume for $n>b$,
	\begin{align}
	\abs{V(n)}\leq \frac{K}{n-b}.\nonumber
	\end{align}
	Then  we have for any $n>n_0>b$,
	\begin{align}\label{thetan}
	\abs{\sum_{l=n_0}^{n}\frac{\sin2\pi\theta(l)}{l-b}}\leq \frac{C(E,K)}{n_0-b},
	\end{align}
	and for any $j=1,2,\cdots,m$
	\begin{align}\label{thetatheta1n}
	\abs{\sum_{l=n_0}^{n}\frac{\sin2\pi\theta(l)\sin2\pi\theta_j(l)}{l-b}}\leq \frac{C(E,A,K)}{n_0-b}.
	\end{align}
\end{lemma}
\begin{proof}[\textbf{Proof.}]
	Let us first show \eqref{thetan}.  It suffices  to prove 
	\begin{align}
	\abs{	\sum_{l=n_0}^{n}\frac{e^{2\pi i\theta(l)}}{l-b}}\leq \frac{C(E,K)}{n_0-b}.\nonumber
	\end{align}
	Indeed, 
	\begin{align}
	&\abs{(e^{2\pi k i}-1)\sum_{l=n_0}^{n}\frac{e^{2\pi i\theta(l)}}{l-b}}
	=\abs{\sum_{l=n_0}^{n}\frac{e^{2\pi i(\theta(l)+k)}}{l-b}-\sum_{l=n_0}^{n}\frac{e^{2\pi i\theta(l)}}{l-b}}\nonumber\\
	=&\abs{\sum_{l=n_0}^{n}\frac{e^{2\pi i(\theta(l)+k)}}{l-b}-\sum_{l=n_0}^{n}\frac{e^{2\pi i\theta(l+1)}}{l-b}+\sum_{l=n_0}^{n}\frac{e^{2\pi i\theta(l+1)}}{l-b}-\sum_{l=n_0}^{n}\frac{e^{2\pi i\theta(l)}}{l-b}}\nonumber\\
	\leq&\abs{\sum_{l=n_0}^{n}\frac{e^{2\pi i(\theta(l)+k)}}{l-b}-\sum_{l=n_0}^{n}\frac{e^{2\pi i\theta(l+1)}}{l-b}}+\abs{\sum_{l=n_0}^{n}\frac{e^{2\pi i\theta(l+1)}}{l-b}-\sum_{l=n_0}^{n}\frac{e^{2\pi i\theta(l)}}{l-b}}\nonumber\\
\leq &\sum_{l=n_0}^{n}\abs{\frac{1-e^{2\pi i(\theta(l+1)-\theta(l)-k)}}{l-b}}+\frac{2}{n_0-b}+\sum_{l=n_0}^{n-1}\abs{\frac{e^{2\pi i\theta(l+1)}}{l-b}-\frac{e^{2\pi i\theta(l+1)}}{l+1-b}}\nonumber\\
= &\sum_{l=n_0}^{n}\abs{\frac{1-e^{2\pi i(\theta(l+1)-\theta(l)-k)}}{l-b}}+\frac{2}{n_0-b}+\sum_{l=n_0}^{n-1}\frac{1}{(l-b)(l+1-b)}.	\label{OKx}
\end{align}
 By  \eqref{bhgx1}, one has that 
\begin{align}
	\theta(l+1)-\theta(l)-k=\frac{O(K)}{(l-b)\sin\pi k}.\label{gj267}
\end{align}
By \eqref{OKx} and \eqref{gj267}, we have
\begin{align}
	\abs{(e^{2\pi k i}-1)\sum_{l=n_0}^{n}\frac{e^{2\pi i\theta(l)}}{l-b}}\label{OK}\leq &\frac{2}{n_0-b}+O\left(\frac{K}{\sin\pi k}\right)\sum_{l=n_0}^{\infty}\frac{1}{(l-b)^2}+\sum_{l=n_0}^{\infty}\frac{1}{(l-b)^2}\\
	\leq &\frac{C(E,K)}{n_0-b}.\nonumber
	\end{align}

	%where \eqref{OK} comes from \eqref{bhgx1}. Then we have \eqref{thetan}.
	
	Clearly,  $$\sin2\pi\theta(l)\sin2\pi\theta_j(l)=\frac{\cos2\pi(\theta(l)-\theta_j(l))-\cos2\pi(\theta(l)+\theta_j(l))}{2}.$$
	By changing $k,\theta(l)$ in the proof of \eqref{thetan} to $k\pm k_j,\theta(l)\pm\theta_j(l)$,  we have \eqref{thetatheta1n}.
\end{proof}

	For simplicity, denote by $K_2=K_2(E,A)$ and $K_1=K_1(E,A)$, we require that 
	$$K_2\gg K_1>0.$$
	
	For $E\neq 0$, we solve the following   system of two equations  for $\theta(n)$ and $\tilde{\theta}(n)$ on $[n_0,\infty)$ with the initial condition $(\theta(n_0),\tilde{\theta}(n_0))=(\theta_0,\tilde{\theta}_0)$ 
\begin{align}
	\begin{cases}\label{es}
		\cot(\pi\theta(n+1)-\pi k)=\cot\pi\theta(n)-\frac{V(n)}{\sin\pi k},\\
		\cot(\pi\tilde{\theta}(n+1)-\pi \tilde{k})=\cot\pi\tilde{\theta}(n)-\frac{V(n)}{\sin\pi \tilde{k}},
	\end{cases}
\end{align}
where 
\begin{align}\label{vc1}
	V(n)=K_1(E,A)\frac{\sin2\pi\theta(n)+\sin2\pi\tilde{\theta}(n)+100}{n-b}.
\end{align}

For $E=0$, instead of solving \eqref{es}, we solve
\begin{align}\label{es1}
	\cot(\pi\theta(n+1)-\pi k)=\cot\pi\theta(n)-\frac{V(n)}{\sin\pi k}
\end{align}
with the initial condition $\theta(n_0)=\theta_0$, where
\begin{align}\label{vc2}
	V(n)=K_1(E,A)\frac{\sin2\pi\theta(n)+100}{n-b}.
\end{align}
%Then \eqref{vxc} follows from \eqref{vc1} and \eqref{vc2}.

\begin{proof}[\textbf{Proof of Proposition \ref{prop} }]
	Substitute \eqref{vc1}  into \eqref{es} (or  \eqref{vc2}  into \eqref{es1}).  Then by solving \eqref{es} (or \eqref{es1}), we obtain $ \theta(n)$ and $\tilde{\theta}(n)$ (or $\theta(n)$) for $n\geq n_0$. For $n\geq n_0$, define $V(n)$ as \eqref{vc1} (or  \eqref{vc2} ).   We finish the construction of $V(n)$. Now we are in the position to prove 
 \eqref{R-100},  \eqref{Rtilde-100} and \eqref{RjC}.
	
\begin{description}
	\item[\textbf{Case 1. \textbf{ $E\neq 0$}}, namely $k\neq \frac{1}{2}$ ] 
\end{description}
Choose a fixed large enough $N$, which will be determined later.

	By \eqref{lr} and \eqref{vc1}, one has
	\begin{align}
		&\ln R(n+N)^2\nonumber\\
		=&\ln R(n)^2-\sum_{l=0}^{N-1}\frac{V(n+l)}{\sin\pi k}\sin2\pi\theta(n+l)+\frac{O(1)}{(n-b)^2}\nonumber\\
		=&\ln R(n)^2-K_1\sum_{l=0}^{N-1}\frac{\sin2\pi\theta(n+l)+\sin2\pi\tilde{\theta}(n+l)+100}{(n-b)\sin\pi k}\sin2\pi\theta(n+l)\nonumber\\
		&+\frac{O(1)}{(n-b)^2}.\label{gj268}
	\end{align}
 
	By \eqref{bhgx1} and \eqref{vc1}, one has that for any $l$ with $0\leq l\leq N-1$, 
	\begin{align}
	\sin2\pi\theta(n+l)&=\sin2\pi\left(\theta(n)+kl+\frac{O(1)}{1+n-b}\right)\nonumber\\&=\sin(2\pi\theta(n)+2\pi kl)+\frac{O(1)}{1+n-b},\label{gj269}
	\end{align}
	and ($k=1-\tilde{k}$)
	\begin{align}
	\sin2\pi\tilde{\theta}(n+l)&=\sin2\pi\left(\tilde\theta(n)+\tilde{k}l+\frac{O(1)}{1+n-b}\right
	)\nonumber\\&=\sin(2\pi\tilde{\theta}(n)-2\pi kl)+\frac{O(1)}{1+n-b}.\label{gj2610}
	\end{align}
 
By trigonometric identities, \eqref{gj268}, \eqref{gj269} and \eqref{gj2610}, we have that 
\begin{align}
	&\ln R(n+N)^2\nonumber\\
	=&\ln R(n)^2-K_1\sum_{l=0}^{N-1}\frac{1-\cos\left(4\pi\theta(n)+4\pi kl\right)}{2(n-b)\sin\pi k}-100K_1\sum_{l=0}^{N-1}\frac{\sin\left(2\pi\theta(n)+2\pi kl\right)}{(n-b)\sin\pi k}\nonumber\\
	&-K_1\sum_{l=0}^{N-1}\frac{\cos\left(2\pi\tilde{\theta}(n)-2\pi\theta(n)-4\pi kl\right)-\cos\left(2\pi\tilde{\theta}(n)+2\pi\theta(n)\right)}{2(n-b)\sin\pi k}+\frac{O(1)}{(n-b)^2},\nonumber
\end{align}
where $O(1)$ depends on $E$  and $A$ through $N$ and $K_1(E,A)$.

Applying Lemma \ref{ki} with large enough $N$, and using  the fact that  $n-b\geq K_2 $, one has that
there exists  $\abs{\delta(n)}\leq \frac{\sin\pi k}{10^4}$ such that 
\begin{align}\label{y1}
	\ln R(n+N)^2=\ln R(n)^2-\frac{K_1N}{2(n-b)\sin\pi k}(1-\cos2\pi\varphi(n)+\delta(n)),
\end{align}
where $\varphi(n)=\tilde{\theta}(n)+\theta(n)$.

By \eqref{bhgx2},  \eqref{vc1} and Lemma \ref{ki} (with trigonometric identities),   we have that there exists $\abs{\varepsilon(n)}<1$ such that
\begin{align}
	&\varphi(n+N)\nonumber\\
	=&\varphi(n)+N+\sum_{l=0}^{N-1}(\sin^2\pi\theta(n+l)+\sin^2\pi\tilde{\theta}(n+l))\frac{V(n+l)}{\pi\sin\pi k}+\frac{O(1)}{(n-b)^2}\nonumber\\
		=&\varphi(n)+N+\frac{O(1)}{(n-b)^2}\nonumber\\
		&
		+K_1\sum_{l=0}^{N-1}(\sin^2\pi\theta(n+l)+\sin^2\pi\tilde{\theta}(n+l))\frac{\sin2\pi\theta(n+l)+\sin2\pi\tilde{\theta}(n+l)+100}{(n-b) \pi\sin\pi k}\nonumber\\
			=&\varphi(n)+N+\frac{O(1)}{(n-b)^2}+100 K_1\sum_{l=0}^{N-1} \frac{1}{(n-b)\pi \sin\pi k}\nonumber\\
				&
			+K_1\sum_{l=0}^{N-1}\frac{\sin2\pi\theta(n+l)+\sin2\pi\tilde{\theta}(n+l)}{(n-b) \pi\sin\pi k}\nonumber\\
		&
		-\frac{K_1}{2}\sum_{l=0}^{N-1}(\cos 2\pi\theta(n+l)+\cos 2\pi\tilde{\theta}(n+l))\frac{\sin2\pi\theta(n+l)+\sin2\pi\tilde{\theta}(n+l)+100}{(n-b) \pi\sin\pi k}\nonumber\\
	\label{y2}=&\varphi(n)+N+\frac{K_1N}{(n-b)\pi\sin\pi k}\left(100+\varepsilon(n)\right).
\end{align}
By \eqref{y1}, \eqref{y2} and applying Lemma \ref{zhongyaoyinli} with $L=N$, one has that  for $n=n_0+mN$, $m=1,2,\cdots$, 
\begin{align}\label{R-1001}
R(n)\leq 1.5\left(\frac{n-b}{n_0-b}\right)^{-100}R(n_0)
\end{align}
   By \eqref{lr} and the fact that $n_0-b$ is large enough,  \eqref{R-1001} implies  for any  $l$ with $ 0\leq l\leq N-1 $ and $n=n_0+m N+l$,
   \begin{align}\label{R-1002}
   R(n)\leq 2\left(\frac{n-b}{n_0-b}\right)^{-100}R(n_0).
   \end{align}
   This implies \eqref{R-100}. 
The proof of  \eqref{Rtilde-100}  follows that of  \eqref{R-100} step by step. We omit the details.

%For $k\neq \frac{1}{2}$ ($E\neq 0$).
%If $E_j=\tilde{E}$, then we have \eqref{Rtilde-100} and therefore \eqref{RjC}. We assume that $E_j\neq \tilde{E}$. 
By  \eqref{lr} and \eqref{vc1}, one obtains
\begin{align}
\ln R_j(n)^2&\leq \ln R_j(n_0)^2-\sum_{l=n_0}^{n-1}\frac{V(l)}{\sin\pi k_j}\sin2\pi\theta_j(l)+\sum_{l=n_0}^{n-1}\frac{O(1)}{(l-b)^2}\nonumber\\
&=\ln R_j(n_0)^2+\frac{O(1)}{n_0-b}-\frac{K_1}{\sin\pi k_j}\sum_{l=n_0}^{n-1}\frac{\sin2\pi\theta(l)+\sin2\pi\tilde{\theta}(l)+100}{l-b}\sin2\pi\theta_j(l).\nonumber\\
&=\ln R_j(n_0)^2+\frac{O(1)}{n_0-b}, \label{g1}
\end{align}
where the last inequality holds by Lemma \ref{keyle}. 
Therefore, for large $n_0-b$, we have \eqref{RjC}.

%Applying \eqref{thetan} and \eqref{thetatheta1n} one obtains \eqref{RjC}.

%Direct computations show that there exist $\abs{\tilde{\delta}(n)}\leq \frac{\sin\pi k}{10^4}$ such that
%\begin{align}\label{y3}
%	\ln \tilde{R}(n+q)^2=\ln \tilde{R}(n)^2-\frac{Cq}{2(n-b)\sin\pi k}(1-\cos2\pi\varphi(n)+\tilde{\delta}(n)),
%\end{align}

%therefore one obtains \eqref{Rtilde-100} by \eqref{y2}, \eqref{y3} and Lemma \ref{zhongyaoyinli}.

\begin{description}
	\item[\textbf{Case 2. \textbf{$k=\frac{1}{2}$ ($E= 0$)}}] 
\end{description}

Since $\sin \pi k=1,$ by \eqref{lr}, \eqref{bhgx2} and \eqref{vc2}, one has 
\begin{align}\label{E0R}
	\ln R(n+2)^2=&\ln R(n)^2-\sum_{j=0}^1V(n+j)\sin2\pi\theta(n+j)+\frac{O(1)}{(n-b)^2}\nonumber\\
	=&\ln R(n)^2-\frac{K_1}{n-b}(1-\cos4\pi\theta(n))+\frac{O(1)}{(n-b)^2},\nonumber\\
		=&\ln R(n)^2-\frac{K_1}{n-b}(1-\cos4\pi\theta(n)+\frac{O(1)}{n-b})
\end{align}
and 
\begin{align}\label{E0theta}
	\theta(n+2)=\theta(n)+1+\frac{K_1}{\pi(n-b)}\left(100-\frac{1}{2}\sin4\pi\theta(n)+\frac{O(1)}{n-b}\right).
\end{align}
By \eqref{E0R},  \eqref{E0theta}  and applying Lemma \ref{zhongyaoyinli} with  $\varphi(n)=\theta(n)$ and $L=1$, we have   \eqref{R-100}.

In this case, the proof of \eqref{RjC} is similar to the proof in Case 1.

\end{proof}

\begin{proof}[\textbf{Proof of Theorem \ref{theorem1} and Theorem \ref{theorem2}}]
As we already mentioned in the introduction, once we  have the generating functions (Proposition \ref{prop}) at hand, proofs of Theorems \ref{theorem1} and \ref{theorem2} follow from the standard  construction   first introduced in \cite{JL19} and further developed in  \cite{LO17,L19stark}. We only give an outline of the  proof here. 

Let $\{N_r\}_{r\in\N}$ be a non-decreasing sequence. 
In the construction of  Theorem \ref{theorem1},  $N_r=M$ for   sufficiently large $r$, where $M$ is the number of all  non-resonant eigenvalues being constructed. 
In  Theorem \ref{theorem2},   $N_r$ goes to infinity arbitrarily slowly. We further assume  $N_{r+1}=N_r+1$ when $N_{r+1}>N_r$. 
The construction is proceeded by inductions. Suppose we construct the potential on $ [0,J_{r-1}]$ before $r-1$ steps. 
At the $r$th step,   we take $N_r$  (non-resonant)  eigenvalues into consideration.  
Our  strategy is to construct  the potential  on $[J_{r-1},J_{r}]=\cup_{l=1}^{N_r}[J_{r-1}+(l-1)T_r,J_{r-1}+lT_r]$  in a piecewise manner (namely at $r$th step, we construct $N_r$ pieces of the potential of size $T_r$ starting at $J_{r-1}$ and ending  at  $J_r=J_{r-1}+N_rT_r$).
For each piece $[J_{r-1}+(l-1)T_r,J_{r-1}+lT_r]$, we apply  Proposition \ref{prop} with $E$ being one eigenvalue (in total we have $N_r$ choices so we obtain $N_r$ pieces) and $A$ being the rest of eigenvalues. 

The main difficulty is to  control  the size of each piece $T_r$. The construction in \cite{JL19,LO17,L19stark} only uses inequalities \eqref{vxc}, \eqref{R-100}, \eqref{Rtilde-100} and \eqref{RjC} to obtain appropriate   $T_r$ and $N_r$.
Therefore Theorems \ref{theorem1} and \ref{theorem2}  follow from  Proposition \ref{prop}. 

\end{proof}

		\section*{Acknowledgments}
	The authors wish to express their gratitude to the anonymous referee, whose comments helped the exposition of this paper. 
	W. Liu was supported by NSF DMS-2000345 and DMS-2052572. K. Lyu was supported by the National Natural Science Foundation of China (11871031).

\end{document}